%% file: modelrotate.tex
\documentclass{article}
\usepackage{array}
\usepackage{graphicx}
\usepackage{float}
\usepackage{enumitem}
\setlist{nolistsep}

\input{environments}
\input{figures}

\input{notation}

\newcommand{\assoc}{assoc}
\newcommand{\assocs}{assocs}
\newcommand{\posedges}{E_P} 
\newcommand{\posgraph}[1][\F]{G=(\F,\posedges)}
\newcommand{\flipedges}{E_G} 
\newcommand{\flipgraph}[1][\F]{G=(\F,\flipedges)}

\newcommand{\citecp}{\cite{DBLP:conf/cp/Wieringa12}}

\author{Siert Wieringa\thanks{The work was financially supported by the Academy of Finland, project 139402} \\\medskip \small{Aalto University, Finland} \\\medskip \small{\texttt{siert.wieringa@aalto.fi}}}
\title{Some notes on model rotation}

\begin{document}
\maketitle

\begin{abstract}
Model rotation is an efficient technique for improving MUS finding algorithms.
In previous work we have studied model rotation as an algorithm that traverses a graph which is induced by the input formula.
This document introduces the notion of blocked edges, which are edges in this graph that can never be traversed.
We show the existence of irredundant CNF formulas in which some clauses are unreachable by model rotation.
Additionally, we prove a conjecture by Belov, Lynce and Marques-Silva.
\end{abstract}

\section{Definitions}

A literal $l$ is a Boolean variable $l=x$ or its negation $l=\lnot x$.
For any literal $l$ it holds that $\lnot \lnot l = l$.
A clause $c=\{l_1, l_2, \cdots, l_{|c|}\}$ is a non-empty set of literals, representing the disjunction $l_1 \vee l_2 \vee \cdots \vee l_{|c|}$.
A propositional logic formula $\F$ is in \emph{Conjunctive Normal Form} (CNF) if it is a conjunction of disjunctions, i.e. a set of clauses.
Throughout this document the word formula always refers to a formula in CNF.

An assignment $a$ is a set of literals such that if $l \in a$ then $\lnot l \notin a$.
If $l \in a$ then it is said that literal $l$ is assigned the value \true,~if $\lnot l \in a$ then $l$ it is said that $l$ assigned value \false.
Assignment $a$ satisfies clause $c$ if there exists a literal $l \in a$ such that $l \in c$.
An assignment satisfies a formula if it satisfies all clauses in the formula.
An assignment $a$ is a \emph{complete assignment} for a formula $\F$ if for all $c \in \F$ and all $l \in c$ either $l \in a$ or $\lnot l \in a$.
A formula $\F$ is \emph{equivalent} to a formula $\F'$, denoted $\F \equiv \F'$, if for all assignments $a$ it holds that $a$ satisfies $\F$ if and only if $a$ satisfies $\F'$.

A formula that has no satisfying assignments is called unsatisfiable.
A formula $\F$ is \emph{minimal unsatisfiable} if it is unsatisfiable and any subformula $\F' \subset \F$ is satisfiable.
A \emph{minimal unsatisfiable subset}\footnote{Alternative names are \emph{minimal unsatisfiable subformula} or \emph{minimal unsatisfiable core} (MUC).} (MUS) of a formula $\F$ is a formula $\F' \subseteq \F$ that is minimal unsatisfiable.
The concept can be generalized to include satisfiable formulas by considering \emph{minimal equivalent subsets} (MESes) \cite{DBLP:conf/cp/BelovJLM12} instead.
A MES $\F' \subseteq \F$ is a formula such that $\F' \equiv \F$ and for all $c \in \F'$ it holds that $\F' \setminus \{c\} \not \equiv \F$ (i.e.\ $\F' \setminus \{c\} \not \models c$).

\begin{definition}{\assoc}
An associated assignment (\assoc) for a clause $c \in \F$ is a complete assignment $a$ for the formula $\F$ that satisfies the formula $\F \setminus \{c\}$ and does not satisfy $c$. Let $A(c,\F)$ be the set of all \assocs~for $c \in \F$.
\end{definition}

Note that iff a clause $c \in \F$ has an \assoc~(i.e. if $A(c,\F) \neq \emptyset$) then $c$ occurs in every MES of $\F$.
Such clauses are often referred to as \emph{critical clauses} or \emph{transition clauses} (e.g \cite{DBLP:conf/sat/SilvaL11}).
In this work we will refer to such clauses as critical clauses, or more explicitly as \emph{clauses for which an \assoc~exists}.

In \cite{DBLP:conf/sat/SilvaL11} a technique called \emph{model rotation} was introduced,
which shortly after was improved to \emph{recursive model rotation} \cite{DBLP:conf/fmcad/BelovM11}.
Model rotation is an algorithm that given an \assoc~for a clause attempts to find \assocs~for other clauses by negating a single literal.

\begin{definition}{Rotation function}Let $\flipfun{a}{l}$ be a function that negates literal $l$ in assignment $a$, i.e.:
$\flipfun{a}{l} = \left ( a \setminus \{ l \} \right ) \cup \{ \lnot l \}$
\end{definition}

The pseudocode for a basic destructive algorithm extended with model rotation is shown in Alg.~\ref{alg:modelrotate}.
The model rotation subroutine can be thought of as an algorithm that traverses a graph, which we call the \emph{flip graph} \citecp.

\newcommand{\Fexample}{\F_{fig\ref{fig:flipgraph}}}
\plotcenter{.}{flipgraph}{The \emph{flip graph} for $\Fexample = \{ \{ x \}, \{\lnot x, y \}, \{ \lnot x, z \}, \{\lnot y, \lnot z\} \}$}{0}

\begin{algorithm}{Destructive algorithm with recursive model rotation}{Given an unsatisfiable formula $\F$:}\label{alg:modelrotate}
  \item $M=\emptyset$
  \item \while~$\F \neq M$
  \item \algtab \pick~a clause $c \in \F \setminus M$
  \item \algtab \ifthen{$\F \setminus \{c\}$ is \sat}
  \item \algtab \algtab $a=$~an assignment satisfying $\F \setminus \{c\}$
  \item \algtab \algtab $M=M\cup\{c\}$
  \item \algtab \algtab $\mrotatefun{c}{a}$
  \item \algtab \elsekw~$\F=\F \setminus \{c\}$
  \item \return~$\F$
  \subroutine{$\mrotatefun{\mbox{\tt clause~} c}{\mbox{~\tt assignment~} a}$}  
  \item \fordo{$l \in c$}
  \item \algtab $a'=\flipfun{a}{\lnot l}$
  \item \algtab \ifthen{exactly one clause $c' \in \F$ is not satisfied by $a'$ and $c' \notin M$}
  \item \algtab\algtab $M=M\cup\{c'\}$
  \item \algtab\algtab $\mrotatefun{c'}{a'}$
\end{algorithm}

\begin{definition}{Flip graph}For a formula $\F$ the flip graph $G=(V,E)$ is a graph which has a vertex for every clause, i.e. $V=\F$.
Each edge $(c_i,c_j) \in E$ is labelled with the set of literals $L(c_i,c_j)$ such that:
\[
L(c_i,c_j)\quad=\quad \{ l \mid l \in c_i \mbox{~and~} \lnot l \in c_j \}
\]
The set of edges $E$ of the flip graph is defined by $(c_i,c_j) \in E$ iff $L(c_i,c_j)\neq \emptyset$
\end{definition}

Even though $(c_i,c_j) \in E$ iff $(c_j,c_i) \in E$ in this work the flip graph is considered to be a directed graph.
This is useful for defining the \emph{rotation edges}.

\begin{definition}{Rotation edges}Given a formula $\F$, let the sets of possible rotation edges\footnote{Note that the set $\posedges$ corresponds to all pairs of clauses $(c_i,c_j)$ on which resolution $c_i \otimes c_j$ can be performed without creating a tautology.} $\posedges$, and guaranteed rotation edges $\flipedges$ be defined as:
\[
\begin{array}{llll}
  \posedges  &=& \{ (c_i,c_j) \mid& c_i,c_j \in \F \mbox{~and~} |L(c_i,c_j)| = 1 \}\smallskip \\
  \flipedges &=& \{ (c_i,c_j) \mid& c_i,c_j \in \F \mbox{~and~} |L(c_i,c_j)| = 1 \mbox{~and for all~} c_k \in \F \\
  &&& \mbox{it holds that~} L(c_i,c_j) \neq L(c_i,c_k) \mbox{~if~} c_k \neq c_j \}
\end{array}
\]
\end{definition}

In Fig.~\ref{fig:flipgraph}~the flip graph for an example formula $\Fexample$ is given.
Because there are no two clauses $c_i,c_j \in \Fexample$ such that $|L(c_i,c_j)|>1$ it holds that the set of possible rotation edges $\posedges$ is equal to the set of all edges $E$ in the flip graph.
However, only the solid edges in the figure belong to the set of guaranteed rotation edges $\flipedges$.
The dotted edges are not in the set $\flipedges$ because the two outgoing edges from vertex $c_1$ have the same label $L(c_1,c_2)=L(c_1,c_3)=\{x\}$.
In \citecp~we prove the following theorem:

\begin{theorem}
\label{th:flipedges}
Let $\F$ be an unsatisfiable formula and $\flipedges$ the set of guaranteed rotation edges it induces.
If $(c_i,c_j) \in \flipedges$ then for any \assoc~$a_i \in A(c_i,\F)$ an assignment $a_j= \flipfun{a_i}{\lnot l}$ such that $L(c_i,c_j)=\{l\}$ is an \assoc~$a_j \in A(c_j,\F)$.
\end{theorem}

This theorem implies that if we find an \assoc~for a clause then model rotation is guaranteed to find an \assoc~for all clauses that are reachable from that clause over edges in $\flipedges$.
It is shown in \citecp~that typical formulas used for benchmarking MUS finding algorithms contain large numbers of guaranteed rotation edges.
This means that an upperbound can be computed on the minimum number of calls to a SAT solver needed by Alg.~\ref{alg:modelrotate},
which is typically much smaller than the number of clauses in the formula.
We used this observation to argue about the strength of model rotation.

\section{Blocked rotation edges}
In \citecp~we defined a subset of possible rotation edges $\flipedges \subseteq \posedges$ on which rotation is guaranteed to succeed.
Here, we discuss the possible existence of edges in $\posedges$ on which rotation is guaranteed to fail.

\begin{definition}{Blocked rotation edge}An edge $(c_i,c_j) \in \posedges$ is \emph{blocked} if for all $a_i \in A(c_i,\F)$ we have $\flipfun{a_i}{\lnot l} \notin A(c_j,\F)$, where $l$ is the literal such that $L(c_i,c_j)=\{l\}$.
\end{definition}

\begin{corollary}If and only if $(c_i,c_j) \in \posedges$ is a blocked edge then $(c_j,c_i) \in \posedges$ is a blocked edge.\end{corollary}

Naturally, any edge $(c_i,c_j) \in \posedges$ such that either $A(c_i,\F)=\emptyset$ or $A(c_j,\F)=\emptyset$ is a blocked edge.
However, we will show that blocked edges may also exist between two critical clauses.

\begin{lemma}\label{lemma:bidirblocks}If $L(c_i,c_j)=\{l\}$ and for some literal $l' \neq l$ it holds that $\F \setminus \{c_i,c_j\} \models l \leftrightarrow l'$ then the edge $(c_i,c_j) \in \posedges$ is blocked.
\end{lemma}

\begin{proof}
For all $a_i \in A(c_i,\F)$ it holds that $\lnot l \in a_i$ and $a_i$ satisfies $\F \setminus \{c_i,c_j\}$, thus $\lnot l' \in a_i$ holds.
But then any assignment $\flipfun{a_i}{\lnot l}$ contains $l$ and $\lnot l'$ and therefore does not satisfy $\F \setminus \{c_i,c_j\}$.
It follows that no such assignment can be an \assoc~for $c_j$.
\end{proof}

Note that Lemma \ref{lemma:bidirblocks} provides a sufficient condition for blocking the edge between two critical clauses $c_i$ and $c_j$, but that this is not a necessary condition.
For example, the lemma can be generalized by replacing the literal $l'$ with any formula $P$ such that $l$ does not occur in $P$ and $\F \setminus \{c_i,c_j\} \models l \leftrightarrow P$.

An interesting observation is that we can create an irredundant formula $\F$ with a clause $c_i \in \F$ such that for all $c_j \in \F$ all edges $(c_i,c_j) \in \posedges$ are blocked.
This means that for this formula model rotation starting at $c_i$ can never find an \assoc~for any other clause, neither can model rotation starting from any other clause result in an \assoc~for clause $c_i$.

\begin{example}\label{ex:blockedrotation}
Consider the following satisfiable irredundant formula $\F$:
\[
\begin{array}{lll@{\quad\quad}lll}
c_0 &=& x \vee y \\
c_1 &=& a \vee \lnot x & c_2 &=& \lnot a \vee x \\
c_3 &=& b \vee \lnot x & c_4 &=& \lnot b \vee x \\
c_5 &=& c \vee \lnot y & c_6 &=& \lnot c \vee y \\
c_7 &=& d \vee \lnot y & c_8 &=& \lnot d \vee y \\
\end{array}
\]
Note that this formula represents four equivalences $a \leftrightarrow x$, $b \leftrightarrow x$, $c \leftrightarrow y$ and $d \leftrightarrow y$.
Together, these make sure that for all $c \in \F$ it holds that the edge $(c_0,c) \in \posedges$ is blocked.
The formula can be made minimal unsatisfiable without breaking this property, for example by adding one clause for each of the three satisfying assignments of this formula:
\[
\begin{array}{lll}
c_9   &=& a \vee b \vee \lnot c \vee \lnot d \vee  x \vee \lnot y \\
c_{10} &=& \lnot a \vee \lnot b \vee c \vee d \vee \lnot x \vee y \\
c_{11} &=& \lnot a \vee \lnot b \vee \lnot c \vee \lnot d \vee \lnot x \vee \lnot y
\end{array}
\]
\end{example}

\section{Proof of a conjecture by Belov et al.}
In \cite{DBLP:journals/aicom/BelovLM12} a conjecture is presented that we prove here.
The conjecture states a property of the \emph{rotation graph}, which was defined alongside the conjecture.
Here we state an equivalent definition for the rotation graph using slightly different notation.

\begin{definition}{Rotation graph}\label{def:rotationgraph}
Let $\F$ be an unsatisfiable formula, and let $Unsat(\F,a)$ be the set of clauses in $\F$ not satisfied by assignment $a$, i.e. $Unsat(\F,a)=\{ c \mid c \in \F \mbox{~and~} c \cap a = \emptyset \}$.
The rotation graph $\mathcal{R_F}=(V_R,E_R)$ is a directed graph which has a vertex for each complete assignment to the variables of $\F$.
There exists an edge $(a,a') \in E_R$ if $a' = \flipfun{a}{\lnot l}$ for some literal $l \in \bigcup Unsat(\F,a)$.
\end{definition}

A \emph{witness assignment}, as mentioned in the following quote, is exactly the same as an \assoc.

\begin{refquote}{Conjecture found in \cite{DBLP:journals/aicom/BelovLM12}}\label{quote:conjecture}Let $\F$ be a minimally unsatisfiable formula, and let $\mathcal{R_F}$ be the rotation graph of $\F$. Then, there exists a witness assignment $v$ such that the traversal of $\mathcal{R_F}$ starting from $v$ visits at least one witness assignment for each clause $c \in \F$.
\end{refquote}

The possible existence of clauses that are connected only through blocked edges in the flip graph, as in Example~\ref{ex:blockedrotation}, does not disprove this conjecture.
This is because the traversal of the rotation graph as defined here may pass through assignments $a$ for which $|Unsat(\F,a)|>1$, i.e. it may perform rotation through assignments that are not an \assoc~for any clause.

\begin{lemma}\label{lemma:inr}
Let $\F$ be an unsatisfiable formula, let $a_i$ be a complete assignment to the variables of $\F$, and let $a_j$ be an \assoc~for some clause $c_j \in \F$, i.e. $a_j \in A(c_j,\F)$.
Either $a_i$ is an \assoc~for clause $c_j$, or there exists a literal $l \in \bigcup Unsat(\F,a_i)$ such that $l \in R$ where $R = a_j \setminus a_i$.
\end{lemma}

\begin{proof}
Let $c_i \in Unsat(\F,a_i)$ such that $c_i \neq c_j$.
Such a clause must exists because $Unsat(\F,a_i)$ is both non-empty and not equal to $\{c_j\}$.
As $a_j$ satisfies $c_i$ and $a_i$ does not, it must hold for some $l \in c_i$ that $l \in a_j$ and $l \notin a_i$, hence $l \in R$.
\end{proof}

\begin{lemma}\label{lemma:belovpath}
Let $\F$ be an unsatisfiable formula, let $a_i$ be a complete assignment to the variables of $\F$, and let $c_i$ be a clause such that $c_i \in Unsat(\F,a_i)$.
For any clause $c_j \in \F$ such that $A(c_j,\F) \neq \emptyset$ there exists a path in the rotation graph starting from the vertex corresponding to assignment $a_i$ to an \assoc~$a_j \in A(c_j,\F)$.
\end{lemma}

\begin{proof}
Let $c_j$ be some clause $c_j \in \F$ such that $A(c_j,\F) \neq \emptyset$.
We will show how to construct a rotation path starting from $a_i$ that is guaranteed to end in an \assoc~for $c_j$.
For some $a_j \in A(c_j,\F)$ let $R=a_j \setminus a_i$.
The path begins at the vertex corresponding to assignment $a=a_i$.
The path is completed when we reach an assignment $a$ that is an \assoc~for $c_j$.
By combining Definition \ref{def:rotationgraph} and Lemma \ref{lemma:inr} we may observe that if $a$ is not an \assoc~for $c_j$ then there exists a literal $l \in R$ such that $(a, a') \in E_G$ for $a'=\flipfun{a}{\lnot l}$.
Hence, the path can proceed from $a$ to $a'$. 
At $a'$ we repeat the previous, i.e. either we find that $a'$ is an \assoc~for $c_j$ or we compute the next step in the path.
As one element is removed from $R$ in every step the path is guaranteed to end in an \assoc~for $c_j$.
\end{proof}

Lemma \ref{lemma:belovpath} states that starting from any complete assignment there exists a path to an \assoc~for any arbitrary critical clause. 
Hence, the conjecture in Quote \ref{quote:conjecture}~must hold.
In fact, we can even strengthen the conjecture to the following corollary.

\begin{corollary}\label{cor:belovstrong}Let $\F$ be an unsatisfiable formula, and let $\mathcal{R_F}$ be the rotation graph of $\F$. 
Starting from any complete assignment to the variables of $\F$ (any vertex in $V_R$), there exists a path in $\mathcal{R_F}$ that visits an \assoc~for every clause $c \in \F$ such that $A(c,\F) \neq \emptyset$.
\end{corollary}

Clearly, a variant of model rotation that may traverse all edges in the rotation graph (called unrestricted EMR in \cite{DBLP:journals/aicom/BelovLM12}) can reach an \assoc~for any critical clause in the input formula, starting from any complete assignment.

\section{Conclusion}
We have shown that it is possible to construct an irredundant, or even minimally unsatisfiable, formula in which some clauses are not reachable at all by model rotation.
Furthermore, we have proven a conjecture stated in \cite{DBLP:journals/aicom/BelovLM12}.

\bibliography{modelrotate}
\end{document}

%% file: environments.tex
\newcounter{customenvs}[section]
\renewcommand{\thecustomenvs}{\thesection.\arabic{customenvs}}

\newcommand{\customenvbegin}[1]{\medskip\refstepcounter{customenvs}\noindent\textbf{#1~\thecustomenvs.~}}
\newcommand{\customenvend}{\medskip}

\newcommand{\namedcustomenvbegin}[2]{\medskip\refstepcounter{customenvs}\noindent\textbf{#1~\thecustomenvs~(#2).~}}

\newenvironment{definition}[1]{\namedcustomenvbegin{Definition}{#1}}{\customenvend}
\newenvironment{refquote}[1]{\namedcustomenvbegin{Quote}{#1}}{\customenvend}

\newenvironment{lemma}{\customenvbegin{Lemma}}{\customenvend}
\newenvironment{corollary}{\customenvbegin{Corollary}}{\customenvend}
\newenvironment{theorem}{\customenvbegin{Theorem}}{\customenvend}

\newenvironment{example}{\renewcommand{\plotbegin}{\begin{figure}[ht]}\customenvbegin{Example}}{\customenvend}

\newenvironment{proof}{\medskip\noindent\textbf{Proof.~}}{\medskip}

\floatstyle{ruled}
\newfloat{algfloat}{t}{algfloattmp}
\floatname{algfloat}{Algorithm}

\newenvironment{algorithm}[2]{\begin{algfloat}\caption{#1}\medskip{#2}\renewcommand{\theenumi}{\arabic{enumi}}\renewcommand{\labelenumi}{\theenumi}\begin{enumerate}}{\end{enumerate}\medskip\end{algfloat}}

\newcommand{\subroutine}[1]{\end{enumerate}\medskip\textbf{subroutine~}{#1}\renewcommand{\theenumi}{\Roman{enumi}}\renewcommand{\labelenumi}{\theenumi}\begin{enumerate}}

\newcommand{\keywfont}[1]{{\tt {#1}}}

\newcommand{\mrotate}{\keywfont{modelRotate}}

\newcommand{\mrotatefun}[2]{\mrotate({#1},{#2})}

\newcommand{\flip}{\keywfont{rotate}}
\newcommand{\flipfun}[2]{\flip({#1},{#2})}

\newcommand{\algtab}{\quad~}
\newcommand{\keyword}[1]{\keywfont{#1}}

\newcommand{\ifthen}[1]{\keyword{if~}{#1}\keyword{~then}}

\newcommand{\fordo}[1]{\keyword{for~}{#1}\keyword{~do}}

\newcommand{\sat}{\textbf{satisfiable}}

\newcommand{\return}{\keyword{return}}
\newcommand{\elsekw}{\keyword{else}}

\newcommand{\while}{\keyword{while}}

\newcommand{\pick}{\keyword{pick}}

%% file: figures.tex
\newcommand{\plotbegin}{\begin{figure}[t]}
\newcommand{\plotend}{\end{figure}}

\newcommand{\plotcenter}[4]{\plotbegin\begin{center}\includegraphics[angle=#4]{#1/#2}\caption{#3}\label{fig:#2}\end{center}\plotend}



%% file: notation.tex
\newcommand{\F}{\mathcal{F}}

\newcommand{\true}{\textbf{true}}
\newcommand{\false}{\textbf{false}}
